\documentclass[11pt,a4paper]{article}
\usepackage{amsmath}
\usepackage{amssymb}
\usepackage{subfigure}
\usepackage{fullpage}
\usepackage{url}
\usepackage{hyperref}
\usepackage[ruled,vlined]{algorithm2e}
\usepackage{rotating}

\newcommand{\adj}[1]{\Gamma(#1)}
\newcommand{\bc}[1]{{\tt bc}[#1]}
\newcommand{\bcp}[1]{{\tt bc}'[#1]}
\newcommand{\sig}[1]{{\sigma}[#1]}
\newcommand{\del}[1]{{\delta}[#1]}

\newcommand{\dis}[1]{{\tt d}[#1]}
\newcommand{\pre}[1]{{\tt P}[#1]}
\newcommand{\reach}[1]{{\tt reach}[#1]}
\newcommand{\ident}[1]{{\tt ident}[#1]}

\newcommand{\union}{\cup}

\newcommand{\bcalg}{\textsc{Bc-Org}}
\newcommand{\bcalgwg}{\textsc{Bc-Reach}}
\newcommand{\bcalgside}{\textsc{Bfs-Side}}

\newtheorem{corr}{Corollary}
\newtheorem{theorem}{Theorem}
\newtheorem{lemma}{Lemma}

\newtheorem{proof}{Pf.}

\title{Shattering and Compressing Networks for Centrality Analysis}
\author{Ahmet Erdem Sar{\i}y\"{u}ce$^{1,2}$, Erik Saule$^1$, Kamer Kaya$^1$, \"{U}mit V. \c{C}ataly\"{u}rek$^{1,3}$\\
  {$^1$ Dept. Biomedical Informatics, The Ohio State University}\\
  {$^2$ Dept. Computer Science and  Engineering, The Ohio State University}\\
  {$^3$ Dept. Electrical and Computer Engineering, The Ohio State
    University}\\
  Email: {\textit\{aerdem,esaule,kamer,umit\}@bmi.osu.edu}\\
}

\begin{document}

\maketitle
\begin{center}
(Previously submitted to ICDM on June 18, 2012)
\end{center}

\begin{abstract}
Who is more important in a network? Who controls the flow between the
nodes or whose contribution is significant for connections?
Centrality metrics play an important role while answering these
questions. The betweenness metric is useful for network analysis and
implemented in various tools. Since it is one of the most
computationally expensive kernels in graph mining, several techniques
have been proposed for fast computation of betweenness centrality. In
this work, we propose and investigate techniques which compress a
network and shatter it into pieces so that the rest of the computation
can be handled independently for each piece. Although we designed and
tuned the shattering process for betweenness, it can be adapted for
other centrality metrics in a straightforward manner. Experimental
results show that the proposed techniques can be a great arsenal to
reduce the centrality computation time for various types of networks.\\
{\bf Keywords:} Betweenness centrality; network analysis; graph
mining; connected components
\end{abstract}

\section{Introduction}\label{sec:int}
Centrality metrics play an important role to successfully detect the
central nodes in various types of networks such as social
networks~\cite{Ediger10,lou10}, biological
networks~\cite{bader2008,Koschutzki08}, power networks~\cite{Jin00},
covert networks~\cite{Krebs02} and decision/action
networks~\cite{simsekb08}. Among these metrics, {\em betweenness} has
always been an intriguing one and it has been implemented in several
tools which are widely used in practice for analyzing networks and
graphs~\cite{Lugowski12,Bader08}. In short, the betweenness
centrality~(BC) score of a node is the sum of the fractions of
shortest paths between node pairs that pass through the node of
interest~\cite{Freeman77}. Hence, it is a measure for the
contribution/load/influence/effectiveness of a node while
disseminating information through a network.

Although betweenness centrality has been proved to be successful for
network analysis, computing betweenness centrality scores of all the
nodes in a network is expensive. The first trivial algorithms for BC
have $\Theta(n^3)$ and $\Theta(n^2)$ time and space complexity,
respectively, where $n$ is the number of nodes in the
network. Considering the size of today's networks, these algorithms
are not practical. Brandes proposed a faster algorithm which has
$\mathcal{O}(nm)$ and $\mathcal{O}(nm + n^2\log{n})$ time complexity
for unweighted and weighted networks, respectively, where $m$ is the
number of node-node interactions in the
network~\cite{brandes2001}. Since the networks in real life are
usually {\it sparse}, $m \approx kn$ for a small $k$,
$\mathcal{O}(nm)$ is much better than $\mathcal{O}(n^3)$. Brandes'
algorithm also has a better, $\mathcal{O}(n + m)$, space complexity
and currently, it is the best algorithm for BC computations. Yet, it
is not fast enough to handle almost 1 billion users of Facebook or 150
million users of Twitter. Several techniques have been proposed to
alleviate the complexity of BC computation for large networks. A set
of works propose using estimated values instead of exact BC
scores~\cite{brandesp07,geisbergerss08}, and others parallelize BC
computations on distributed memory
architectures~\cite{Lichtenwalter11}, multicore
CPUs~\cite{bader2008,bader06,madduri2009}, and
GPUs~\cite{Shi11,Pande11,Jia11}.

In this work, we propose a set of techniques which compress a network
and break it into pieces such that the BC scores of two nodes in two
different pieces can be computed independently, and hence, in a more
efficient manner. Although we designed and tuned these techniques for
standard, shortest-path vertex-betweenness centrality, they can be
modified for other path-based centrality metrics such as {\em
  closeness} or other BC variants such as {\em edge betweenness} and
{\em group betweenness}~\cite{Brandes08}. Similarly, although we are
interested in unweighted undirected networks in this paper, our
shattering techniques are valid also for weighted directed
networks. Experimental results show that proposed techniques are very
effective and they can be a great arsenal to reduce the computation in
practice.

The rest of the paper is organized as follows: In
Section~\ref{sec:bac}, an algorithmic background for betweenness
centrality is given. The proposed shattering and compression
techniques are explained in
Section~\ref{sec:sha}. Section~\ref{sec:exp} gives experimental
results on various kinds of networks, and Section~\ref{sec:con}
concludes the paper.

\section{Background}\label{sec:bac}

Let $G = (V,E)$ be a network modeled as a graph with $n$ vertices and
$m$ edges where each node in the network is represented by a vertex in
$V$, and an interaction between two nodes is represented by an edge in
$E$. We assume that $\{v,v\} \notin E$ for any $v \in V$, i.e., $G$ is
{\em loop} free. Let $\adj{v}$ be the set of vertices which are
connected to $v$.

A graph $G' = (V', E')$ is a {\em subgraph} of $G$ if $V' \subseteq V$
and $E' \subseteq E$. A \textit{path} is a vertex sequence such that
there exists an edge between consecutive vertices. A path between two
vertices $s$ and $t$ is denoted by $s \leadsto t$.  Two vertices $u$
and $v$ in $V$ are \textit{connected} if there is a path from $u$ to
$v$.  If $u$ and $v$ are connected for all $u,v \in V$ we say $G$ is
\textit{connected}. If $G$ is not connected, then it is {\em
  disconnected} and each maximal connected subgraph of $G$ is a {\em
  connected component}, or a component, of $G$.

Given a graph $G = (V,E)$, an edge $e \in E$ is a {\em bridge}
if $G-e$ has more connected components than $G$ where $G-e$ is
obtained by removing $e$ from $E$. Similarly, a vertex $v \in V$ is
called an {\em articulation vertex} if $G-v$ has more connected
components than $G$ where $G-v$ is obtained by removing $v$ and its
edges from $V$ and $E$, respectively. If $G$ is connected and it does
not contain an articulation vertex we say $G$ is {\em biconnected}. A
maximal biconnected subgraph of $G$ is a {\em biconnected
component}. Hence, if $G$ is biconnected it has only one biconnected
component which is $G$ itself.

$G = (V,E)$ is a {\em clique} if and only if $\forall u,v \in V,
\{u,v\} \in E$. The subgraph {\em induced by} a subset of vertices $V'
\subseteq V$ is $G' = (V', E' = \{V' \times V'\} \cap E)$. A vertex $v
\in V$ is a {\em side vertex} of $G$ if and only if the subgraph of
$G$ induced by $\adj{v}$ is a clique. Two vertices $u$ and $v$ are
    {\em identical} if and only if $\adj{u} = \adj{v}$. $v$ is a
    {\em{degree}-1} vertex if and only if $|\adj{v}| = 1$.

\subsection{Betweenness Centrality}

The betweenness metric is first defined by Freeman in Sociology to
quantify a person's importance on other people's communication in a
social network~\cite{Freeman77}. Given a graph $G$, let $\sigma_{st}$
be the number of shortest paths from a source $s \in V$ to a target $t
\in V$. Let $\sigma_{st}(v)$ be the number of such $s \leadsto t$
paths passing through a vertex $v \in V, v\neq s,t$. Let the {\em pair
  dependency} of $v$ to $s,t$ pair be the fraction $\delta_{st}(v) =
\frac{\sigma_{st}(v)}{\sigma_{st}}$. The betweenness centrality of $v$
is defined as
\begin{equation}
	\bc{v} = \sum_{s \neq v \neq t \in V} \delta_{st}(v). \label{eq:first}
\end{equation}	

Since there are $\mathcal{O}(n^2)$ pairs in $V$, one needs
$\mathcal{O}(n^3)$ operations to compute $\bc{v}$ for all $v \in V$ by
using~(\ref{eq:first}). Brandes reduced this complexity and proposed
an $\mathcal{O}(mn)$ algorithm for unweighted
networks~\cite{brandes2001}. The algorithm is based on the
accumulation of pair dependencies over target vertices. After
accumulation, the dependency of $v$ to $s \in V$ is
\begin{equation} \label{eq:pair}
\delta_{s}(v) = \sum_{t \in V} \delta_{st}(v).
\end{equation}

Let ${\tt P}_s(u)$ be the set of $u$'s predecessors on the shortest
paths from $s$ to all vertices in $V$. That is,
$${\tt P}_s(u) = \{v \in V: \{u,v\} \in E, {\tt d}_s(u) = {\tt d}_s(v)
+ 1\}$$ where ${\tt d}_s(u)$ and ${\tt d}_s(v)$ are the shortest
distances from $s$ to $u$ and $v$, respectively. ${\tt P}_s$
defines the {\em shortest paths graph} rooted in $s$. Brandes observed
that the accumulated dependency values can be computed recursively as
\begin{equation}
\delta_{s}(v) = \sum_{u: v \in {\tt P}_s(u)} \frac{\sigma_{sv}}{\sigma_{su}} \left(1 + \delta_{s}(u)\right). \label{eq:recursion}
\end{equation}

To compute $\delta_{s}(v)$ for all $v \in V \setminus \{s\}$, Brandes'
algorithm uses a two-phase approach. First, to compute $\sigma_{sv}$
and ${\tt P}_s(v)$ for each $v$, a breadth first search~(BFS) is
initiated from $s$. Then in a {\it back propagation} phase,
$\delta_{s}(v)$ is computed for all $v \in V$ in a bottom-up manner by
using~(\ref{eq:recursion}).  Each phase takes a linear time, and hence
this process takes $\mathcal{O}(m)$ time. Since there are $n$ source
vertices and the phases are repeated for each source vertex, the total
complexity of the algorithm is $\mathcal{O}(mn)$. The pseudo-code of
Brandes' betweenness centrality algorithm is given in
Algorithm~\ref{alg:brandes}.

\begin{algorithm}
\DontPrintSemicolon
\SetKwComment{tcp}{$\triangleright$}{}
\small
\caption{\bcalg}
\label{alg:brandes}
\KwData{${G = (V,E)}$}  
  $\bc{v} \leftarrow 0, \forall v \in V$ \;
  \For{{\bf each} $s \in V$} {    
    $S \leftarrow$ empty stack\;
    $Q \leftarrow$ empty queue \;    

    $\pre{v} \leftarrow$ empty list$, \forall v \in V$ \;
    $\sig{v} \leftarrow 0, \forall v \in V$ \;
    $\dis{v} \leftarrow -1, \forall v \in V$ \;

    $Q$.push($s$);\ \  $\sig{s} \leftarrow 1$;\ \  $\dis{s} \leftarrow 0$ \;

    \tcp{Phase $1$: BFS from $s$}    
    \While{$Q$ is not empty} {
      $v \leftarrow Q$.pop() \;
      $S$.push($v$) \;
      \For{{\bf all} $w \in \adj{v}$}{
        \If{$\dis{w} < 0$}{
          $Q$.push($w$) \;
          $\dis{w} \leftarrow \dis{v} + 1$ \;
        }
        \If{$\dis{w} = \dis{v} + 1$}{
          $\sig{w} \leftarrow \sig{w} + \sig{v}$ \;
          $\pre{w}$.push($v$) \;
        }
      }
    }

    \tcp{Phase $2$: Back propagation}
    $\del{v} \leftarrow 0, \forall v \in V$ \;

    \While{$S$ is not empty}{
      $w \leftarrow S$.pop() \;
      \For{$v \in P[w]$}{
        $\del{v} \leftarrow \del{v} + \frac{ \sigma[v] }{ \sig{w}
        } (1+ \del{w})$ \;
      }
      \If{ $w \neq s$}{
        $\bc{w} \leftarrow \bc{w} + \del{w}$ \;
      }
    }
  }
  \Return{{\tt bc}} \;
\end{algorithm}

\section{Shattering and Compressing Networks}
\label{sec:sha}

\subsection{Principle}

Let us start with a simple example: Let $G = (V,E)$ be a binary tree
with $n$ vertices hence $m = n - 1$. If Brandes' algorithm is used the
complexity of computing the BC scores is $\mathcal{O}(n^2)$. However,
by using a structural property of $G$, one can do much better: there
is exactly one path between each vertex pair in $V$. Hence for a
vertex $v \in V$, $\bc{v}$ is the number of (ordered) pairs
communicating via $v$, i.e.,
$$\bc{v} = 2 \times \left((l_vr_v) + (n - l_v - r_v - 1) (l_v +
r_v)\right)$$ where $l_v$ and $r_v$ are the number of vertices in the
left and the right subtrees of $v$, respectively. Since $l_v$ and
$r_v$ can be computed in linear time for all $v \in V$, this approach,
which can be easily extended to an arbitrary tree, takes only
$\mathcal{O}(n)$ time.

As mentioned in Section~\ref{sec:int}, computing BC scores is an
expensive task. However, as the above example shows, some structural
properties of the networks can be effectively used to reduce the
complexity. Unfortunately, an $n$-fold improvement on the execution
time is usually not possible since real-life networks rarely have a
tree-like from.  However, as we will show, it is still possible to reduce
the execution time by using a set of special vertices and edges.

\begin{figure}[htbp]
\begin{center}
\subfigure[A toy social network with various types of vertices: Arthur
is an articulation vertex, Diana is a side vertex, Jack and Martin are
degree-1 vertices, and Amy and May are identical
vertices.]{\includegraphics[width=0.43\textwidth]{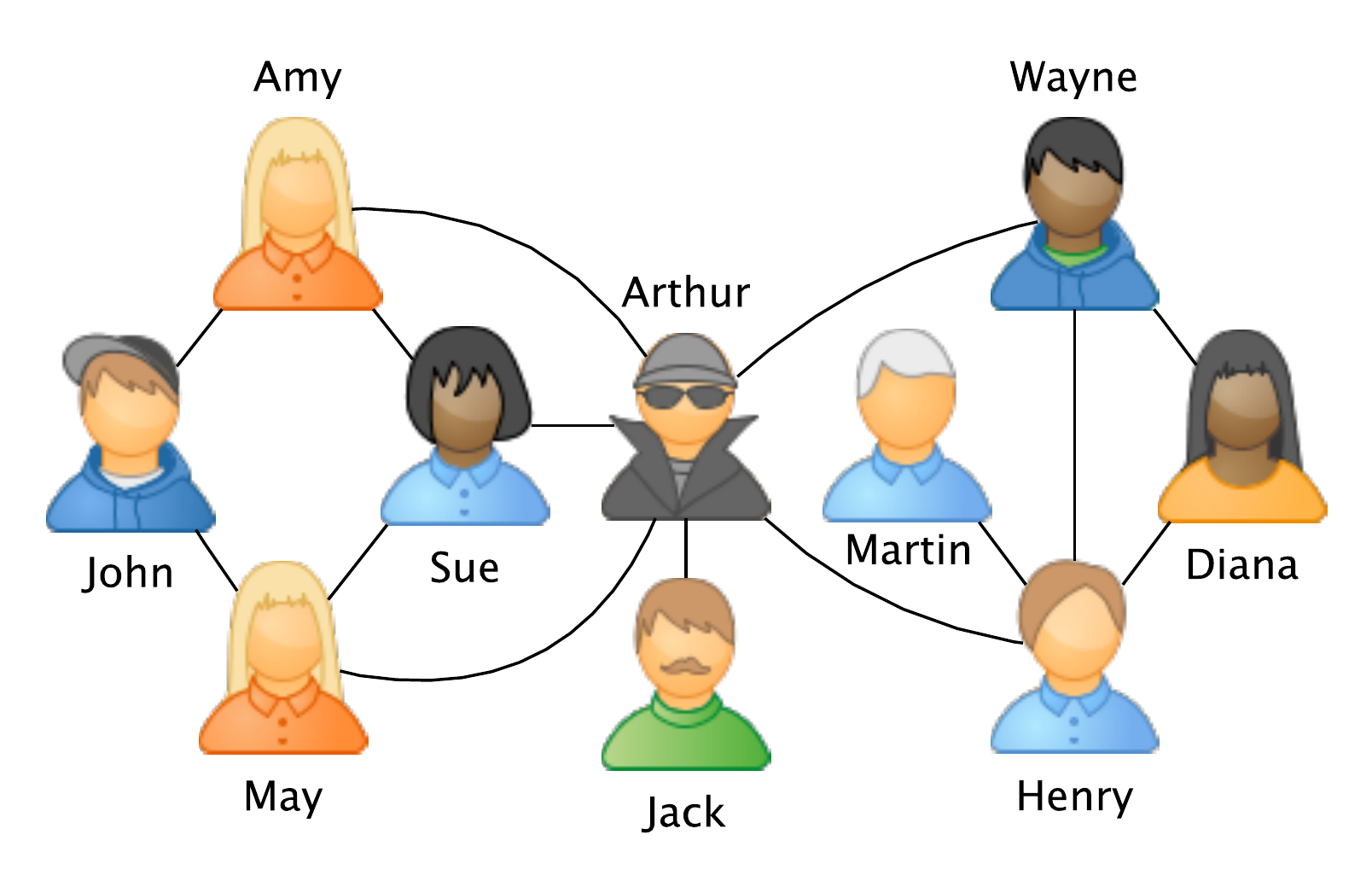}}
\hspace*{3ex}
\subfigure[The network shattered at Arthur to three
components.]{\includegraphics[width=0.43\textwidth]{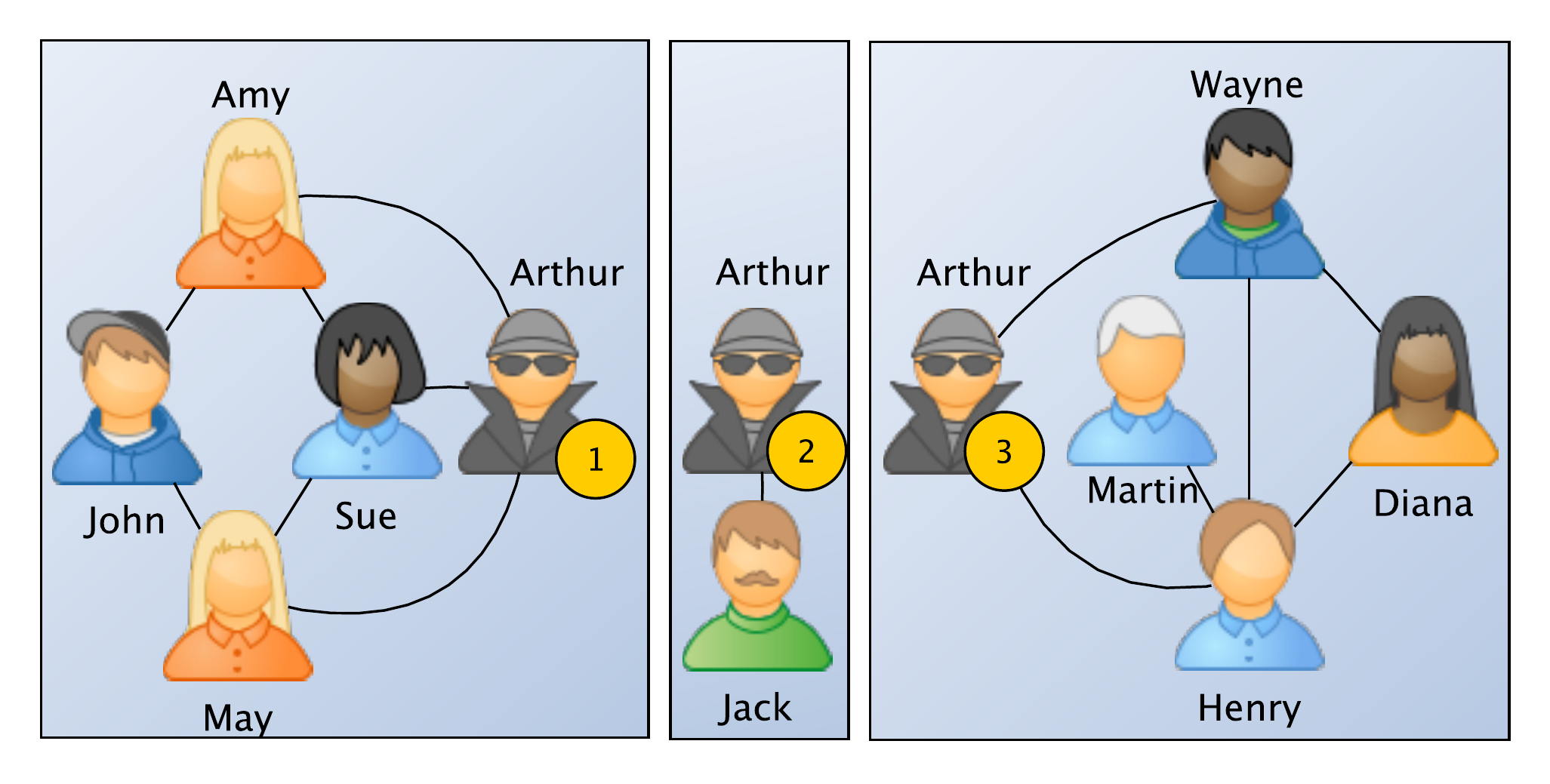}}
\caption{A toy social network and its shattered form due to an
  articulation vertex.}
\label{fig:social}
\end{center}
\end{figure}

Consider the toy graph $G$ of a social network given in
Figure~\ref{fig:social}.(a). Arthur is an articulation
vertex in $G$ and  he is responsible from all inter-communications among
three (biconnected) components as shown in
Figure~\ref{fig:social}.(b). Let $s$ and $t$ be two vertices which lie
in different components.  For all such $s,t$ pairs, the pair
dependency of Arthur is $1$. Since shattering the graph at Arthur removes all $s
\leadsto t$ paths, one needs to keep
some information to correctly update the BC scores of the vertices
inside each component, and this can be achieved creating local copies
of Arthur in each component.

In addition to shattering a graph $G$ into pieces, we investigated
three compression techniques using degree-1 vertices, side vertices,
and identical vertices. These vertices have special properties: All
degree-1 and side vertices always have a zero BC score since they
cannot be on a shortest path unless they are one of the
endpoints. Furthermore, $\bc{u}$ is equal to $\bc{v}$ for two
identical vertices $u$ and $v$. By using these observations, we will
formally analyze the proposed shattering and compression
techniques and provide formulas to compute the BC scores correctly.

We apply our techniques in a preprocessing phase as follows: Let $G =
G_0$ be the initial graph, and $G_\ell$ be the graph after the $\ell$th
shattering/compression operation.  Without loss of generality, we
assume that the initial graph $G$ is connected. The $\ell+1$th
operation modifies a single connected component of $G_\ell$ and
generates $G_{\ell+1}$. The preprocessing phase then checks if
$G_{\ell+1}$ is amenable to further modification, and if this is the
case, it continues. Otherwise, it terminates and the final BC
computation begins.

\subsection{Shattering Graphs}

To correctly compute the BC scores after shattering a graph, we assign
a ${\tt reach}$ attribute to each vertex. Let $G = (V,E)$. Let $v'$ be
a vertex in the shattered graph $G'$ and $C'$ be its component. Then
$\reach{v'}$ is the number of vertices of $G$ which are represented by
$v'$ in $C'$. For instance in Figure~\ref{fig:social}.(b),
$\reach{Arthur_{3}}$ is 6 since Amy, John, May, Sue, Jack, and Arthur
have the same shortest path graphs in the right component.  At the
beginning, we set $\reach{v} = 1$ for all $v \in V$.

\subsubsection{Shattering with articulation vertices}

Let $u'$ be an articulation vertex detected in a connected component
$C \subseteq G_\ell$ after the $\ell$th operation of the preprocessing
phase. We first shatter $C$ into $k$ (connected) components $C_i$ for
$1 \leq i \leq k$ by removing $u'$ from $G_\ell$ and adding a local
copy $u'_i$ of $u'$ to each component by connecting it to the same
vertices $u$ was connected.
The ${\tt reach}$ values for each local copy is set as
\begin{equation}
\reach{u'_i} = \sum_{v' \in C \setminus C_i}{\reach{v'}}
\label{eq:art:update}
\end{equation}
for $1 \leq i \leq k$.  We will use ${\bf org}(v')$ to denote the
mapping from $V'$ to $V$, which maps a local copy $v' \in V'$ to the corresponding
original copy in $V$.

For each component $C$, formed at any time of the preprocessing phase,
a vertex $s \in V$ has exactly one {\em representative} $u' \in C$
such that $\reach{u'}$ is incremented by one due to $s$.
This vertex is denoted as ${\bf rep}(C,s)$. Note that each copy is a
representative of its original. And if ${\bf rep}(C,s) = u'$ and $t' \neq
u'$ is another vertex in $C$ then ${\bf org}(u')$ is on all $s \leadsto
{\bf org}(t')$ paths in $G$.

\begin{algorithm}
\DontPrintSemicolon
\SetKwComment{tcp}{$\triangleright$}{}
\small
\caption{\bcalgwg}
  \label{algo:bcwg}
\KwData{${G' = (V', E')}$ and ${\tt reach}$}
  $\bcp{v} \leftarrow 0, \forall v \in V'$ \;
  \For{{\bf each} $s \in V'$}{
    $\cdots$ \tcp{same as \bcalg}
    \While{$Q$ is not empty}{
      $\cdots$ \tcp{same as \bcalg}
    }
    \lnl{ln:middle}$\del{v} \leftarrow \reach{v} - 1, \forall v \in V'$ \;
    \While{$S$ is not empty}{
      $w \leftarrow S$.pop() \;     
      \For{$v \in \pre{w}$}{
        \lnl{ln:depen}$\del{v} \leftarrow \del{v} + \frac{ \sig{v} }{ \sig{w}} (1+ \del{w})$ \;
      }
      \If{ $w \neq s$}{
        \lnl{ln:source}$\bcp{w} \leftarrow \bcp{w} + (\reach{s} * \del{w})$ \;
      }
    }
  }
  \Return {\tt bc'}

\end{algorithm}

Algorithm~\ref{algo:bcwg} computes the BC scores of the vertices in a
shattered graph. Note that the only difference w.r.t. $\bcalg$ are
lines~\ref{ln:middle} and~\ref{ln:source}. Furthermore, if $\reach{v}
= 1$ for all $v \in V$ the algorithms are equivalent. Hence the worst
case complexity of $\bcalgwg$ is also $\mathcal{O}(mn)$ for a graph
with $n$ vertices and $m$ edges.

Let $G = (V,E)$ be the initial graph with $n$ vertices and $G' =
(V',E')$ be the shattered graph after preprocessing.  Let ${\tt bc}$
and ${\tt bc'}$ be the BC scores computed by $\bcalg(G)$ and
$\bcalgwg(G')$, respectively. We will prove that
\begin{equation}
\bc{v} = \sum_{v' \in V' | {\bf org}(v') = v}{\tt bc}'[v'], \label{eq:toprove}
\end{equation} 
when the graph is shattered at articulation vertices. That is,
$\bc{v}$ is distributed to $\bcp{v'}$s where $v'$ is an arbitrary copy
of $v$. Let us start with two lemmas.

\begin{lemma}\label{lem:one}
Let $u,v,s$ be vertices such that all $s \leadsto v$ paths
contain $u$. Then,
$$\delta_{s}(v) = \delta_{u}(v).$$
\end{lemma}
\begin{proof}
For any target vertex $t$, if $\sigma_{st}(v)$ is positive then
$$\delta_{st}(v) = \frac{\sigma_{st}(v)}{\sigma_{st}} =
\frac{\sigma_{su}\sigma_{ut}(v)}{\sigma_{su}\sigma_{ut}} =
\frac{\sigma_{ut}(v)}{\sigma_{ut}} = \delta_{ut}(v)$$ since all $s
\leadsto t$ paths are passing through $u$. According
to~\eqref{eq:pair}, $\delta_{s}(v) = \delta_{u}(v)$.
\end{proof}

\begin{lemma}\label{lem:two}
For any vertex pair $s,t \in V$, there exists exactly one component
$C$ of $G'$ which contains a copy of $t$ which is not the same vertex
as the representative of $s$ in $C$.
\end{lemma}
\begin{proof}
Given $s,t \in V$, the statement is true for the initial (connected)
graph $G$ since it contains one copy of each vertex. Assume that it is
also true after $\ell$th shattering and let $C$ be this
component. When $C$ is further shattered via $t$'s copy, all but one
newly formed (sub)components contains a copy of $t$ as the
representative of $s$. For the remaining component $C'$, ${\bf
rep}(C',s) = {\bf rep}(C,s)$ which is not a copy of $t$. 

For all components other than $C$, which contain a copy $t'$ of $t$,
the representative of $s$ is $t'$ by the inductive assumption. When
such components are further shattered, the representative of $s$ will
be again a copy of $t$. Hence the statement is true for $G_{\ell
+ 1}$, and by induction, also for $G'$.
\end{proof}

\begin{theorem} \label{thm:art}
Eq.~\ref{eq:toprove} is correct after shattering $G$ with articulation
vertices. 
\end{theorem}
\begin{proof}
Let $C$ be a component of $G'$, $s', v'$ be two vertices in $C$, and
$s, v$ be the corresponding original vertices in $V$,
respectively. Note that $\reach{v'} - 1$ is the number of vertices $t
\neq v$ such that $t$ does not have a copy in $C$ and $v$ lies on all
$s \leadsto t$ paths in $G$. For all such vertices, $\delta_{st}(v) =
1$, and the total dependency of $v'$ to all such $t$ is $\reach{v'} - 1$.
When the BFS is started from $s'$, line~\ref{ln:middle} of $\bcalgwg$
 initiates $\del{v'}$ with this value and computes the final
 $\del{v'} = \delta_{s'}(v')$. This is exactly the same dependency
 $\delta_s(v)$ computed by $\bcalg$.  

Let $C$ be a component of $G'$, $u'$ and $v'$ be two vertices in $C$,
 and $u = {\bf org}(u')$, $v = {\bf org}(v')$. According to the above
 paragraph, , $\delta_u(v) = \delta_{u'}(v')$ where $\delta_u(v)$ and
 $\delta_{u'}(v')$ are the dependencies computed by $\bcalg$ and
 $\bcalgwg$, respectively. Let $s \in V$ be a vertex, s.t. ${\bf
 rep}(C,s) = u'$. According to Lemma~\ref{lem:one}, $\delta_s(v) =
 \delta_u(v) = \delta_{u'}(v')$. Since there are $\reach{u'}$ vertices
 represented by $u'$ in $C$, the contribution of the BFS from $u'$ to
 the BC score of $v'$ is $\reach{u'} \times \delta_{u'}(v')$ as shown
 in line~\ref{ln:source} of $\bcalgwg$. Furthermore, according to
 Lemma~\ref{lem:two}, $\delta_{s'}(v')$ will be added to exactly one
 copy $v'$ of $v$. Hence,~\eqref{eq:toprove} is correct.
\end{proof}

\subsubsection{Shattering with bridges}

Although the existence of a bridge implies the existence of an
articulation vertex, handling bridges are easier and only requires the
removal of the bridge. We embed this operation to the preprocessing
phase as follows: Let $G_\ell$ be the shattered graph obtained after
$\ell$ operations, and let $\{u',v'\}$ be a bridge in a component
$C$ of $G_\ell$. Hence, $u'$ and $v'$ are both articulation vertices. Let $u =
{\bf org}(u')$ and $v = {\bf org}(v')$. A bridge removal operation is
similar to a shattering via an articulation vertex, however, no new
copies of $u$ or $v$ are created. Instead, we let $u'$ and $v'$ act as
a copy of $v$ and $u$.

Let $C_u$ and $C_v$ be the components formed after removing
edge $\{u',v'\}$ which contain $u'$ and $v'$, respectively. Similar
to~\eqref{eq:art:update}, we add
\begin{equation}
\sum_{w \in C_v}{\reach{w}}\notag \mbox{\ \ \ \ and\ \ \ \ } \sum_{w
  \in C_u}{\reach{w}}\label{eq:bri:update}
\end{equation}
to $\reach{u'}$ and $\reach{v'}$, respectively, to make $u'$~($v'$) as the
representative of all vertices in $C_u$~($C_v$).

After removing the bridge and updating the ${\tt reach}$ array,
Lemma~\ref{lem:two} is not true: there cannot be a component which
contain a representative of $u$~($v$) and a copy of $v$~($u$)
anymore. Hence, $\delta_v(u)$ and $\delta_u(v)$ will not be added to
any copy of $u$ and $v$, respectively, by $\bcalgwg$. To alleviate
this, we add
\begin{align}
\delta_{v'}(u') &= \left(\left(\sum_{w \in C_u}{\reach{w}}\right) - 1\right)\sum_{w \in C_v}{\reach{w}}\notag,\\
\delta_{u'}(v') &= \left(\left(\sum_{w \in C_v}{\reach{w}}\right) - 1\right)\sum_{w \in C_u}{\reach{w}}\notag
\end{align}
to ${\tt bc'}[u']$ and ${\tt bc'}[v']$, respectively. Note that Lemma~\ref{lem:two}
is true for all other vertex pairs.

\begin{corr}
Eq.~\ref{eq:toprove} is correct after shattering $G$ with articulation
vertices and bridges.
\end{corr}

\subsection{Compressing Graphs}

Although, the compression techniques do not reduce the number of
connected components, they reduce the number of vertices and edges in
a graph. Since the complexity of Brandes' algorithm is
$\mathcal{O}(mn)$, a reduction on $m$ and/or $n$ will help to
reduce the execution time of the algorithm.

\subsubsection{Compression with degree-1 vertices}

Let $G_\ell$ be the graph after $\ell$ shattering operations, and let
$u' \in C$ be a degree-1 vertex in a component $C$ of $G_\ell$ which
is only connected to $v'$. Removing a degree-1 vertex from a graph is
the same as removing the bridge $\{u',v'\}$ from $G_\ell$. But this
time, we reduce the number of vertices and the graph is
compressed. Hence, we handle this case separately and set $G_{\ell +
1} = G_\ell - u'$. The updates are the same with the bridge
removal. That is, we add $\reach{u'}$ to $\reach{v'}$ and increase
${\tt bc'}[u']$ and ${\tt bc'}[v']$, respectively, with
\begin{align}
\delta_{v'}(u') &= \left(\reach{u'} - 1\right)\sum_{w \in C  \setminus \{u'\}}{\reach{w}},\notag\\
\delta_{u'}(v') &= \left(\left(\sum_{w \in C  \setminus \{u'\}}{\reach{w}}\right) - 1\right) \reach{u'}.\notag
\end{align}

\begin{corr}
Eq.~\ref{eq:toprove} is correct after shattering $G$ with articulation
vertices and bridges, and compressing it with degree-1 vertices.
\end{corr}

\subsubsection{Compression with side vertices}
Let $G_\ell$ be the graph after $\ell$ shattering and compression
operations, and let $u'$ be a side vertex in a component $C$ of
$G_\ell$. Since $\adj{u'}$ is a clique, there is no shortest path
passing through $u'$. That is, $u'$ is always on the sideways. Hence,
we can remove $u'$ from $G_\ell$ by only compensating the effect of
the shortest $s' \leadsto t'$ paths where $u'$ is either $s'$ or
$t'$. To alleviate this, we initiate a BFS from $u'$ as given in
Algorithm~\ref{alg:bcalgside}, which is similar to the ones in
$\bcalgwg$. The only difference between $\bcalgside$ and a BFS of
$\bcalgwg$ is an additional line~\ref{ln:bcadd}.

\begin{algorithm}
\DontPrintSemicolon \SetKwComment{tcp}{$\triangleright$}{} \small
\caption{\bcalgside}
\KwData{${G_\ell = (V_\ell, E_\ell)}$, a side vertex $s$, ${\tt reach}$, and ${\tt bc'}$}

    $\cdots$ \tcp{same as the BFS init. in \bcalgwg}
    \While{$Q$ is not empty}{
      $\cdots$ \tcp{same as BFS in \bcalgwg}
    }
     $\del{v} \leftarrow \reach{v} - 1, \forall v \in V_\ell$ \;
    \While{$S$ is not empty}{
      $w \leftarrow S$.pop() \;     
      \For{$v \in \pre{w}$}{
        $\del{v} \leftarrow \del{v} + \frac{ \sig{v} }{ \sig{w}} (1+ \del{w})$ \;
      }
      \If{ $w \neq s$}{
        \lnl{ln:bcaddorg}$\bcp{w} \leftarrow \bcp{w} + (\reach{s} * \del{w}) + $ \;
        \lnl{ln:bcadd}$\ \ \ \ \ \ \ \ \ \ \ (\reach{s} * (\del{w} - (\reach{w} - 1))$ \;
      }
    }
    \Return {\tt bc'}
  \label{alg:bcalgside}
\end{algorithm}

Removing $u'$ affects three types of dependencies: 
\begin{enumerate}
\item Let $s \in V$ be a vertex s.t. ${\bf rep}(C,s) = u'$ and let $v'$ be a
vertex in $C$ where $v = {\bf org}(v')$. Due to
Lemma~\ref{lem:two}, when we remove $u'$ from $C$, $\delta_s(v) =
\delta_{u'}(v')$ cannot be added anymore to any copy of
$v$. Line~\ref{ln:bcaddorg} of $\bcalgside$ solves this problem and
adds the necessary values to ${\tt bc'}(v')$.

\item Let $s \in V$ be a vertex s.t. ${\bf rep}(C, s) = v' \neq
u'$. If we remove $u'$ from $C$, due to Lemma~\ref{lem:two},
$\delta_s(u) = \delta_{v'}(u')$ will not be added to any copy of
$u$. Since, $u'$ is a side vertex, $\delta_{v'}(u') = \reach{u'} -
1$. Since there are $\sum_{v' \in C - u'}{\reach{v'}}$ vertices which
are represented by a vertex in $C - u'$, we add
$$(\reach{u'} - 1)\sum_{v' \in C - u'}{\reach{v'}}$$ to ${\tt bc}'[u']$
after removing $u'$ from $C$. 

\item Let $v', w'$ be two vertices in $C$ different than $u'$, and
  $v,w$ be the corresponding original vertices.  Although both
  vertices will keep existing in $C - u'$, since $u'$ will be removed,
  $\delta_{v'}(w')$ will be $\reach{u'}\times \delta_{v'u'}(w')$ less
  than it should be. For all such $v'$, the aggregated dependency will be
  $$\sum_{v' \in C, v' \neq w'} \delta_{v'u'}(w') = \delta_{u'}(w') -
  (\reach{w'} - 1),$$ since none of the $\reach{w'} - 1$ vertices
  represented by $w'$ lies on a $v' \leadsto u'$ path and
  $\delta_{v'u'}(w') = \delta_{u'v'}(w')$. The same dependency appears
  for all vertices represented by $u'$. Line~\ref{ln:bcadd} of
  $\bcalgside$ takes into account all these dependencies.
\end{enumerate}

\begin{corr}
Eq.~\ref{eq:toprove} is correct after shattering $G$ with articulation
vertices and bridges, and compressing it with degree-1 and side vertices.
\end{corr}

\subsubsection{Compression with identical vertices}

When two vertices in $G$ are identical, all of their pair
dependencies, source dependencies, and BC scores are the same. Hence,
it is possible to combine these vertices and avoid extra
computation. We distinguish 2 different types of identical
vertices. Vertices $u$ and $v$ are type-I identical if and only if
$\Gamma(u) = \Gamma(v)$. Vertices $u$ and $v$ are type-II identical if
and only if $ \Gamma(u) \union \{ u \} = \Gamma(v) \union \{ v \}$.

To handle this, we assign ${\tt ident}$ attribute to each
vertex. ${\tt ident}(v')$ denotes the number of vertices in $G$ that
are identical to $v'$ in $G'$. Initially, $\ident{v'}$ is
set to $1$ for all $v \in V$.

Let ${\cal I} \subset V$ be a set of identical vertices. We remove all
vertices $u' \in {\cal I}$ from $G$ except one of them. Let $v'$ be
this remaining vertex.  We increase $\ident{v'}$ by $|{\cal I}| - 1$,
and keep a list of ${\cal I} \backslash \{v'\}$'s associated with $v'$.

When constructing the BFS graph, the number of paths $\sig{w}$ is
updated incorrectly for an edge $\{v,w\}$ when $v$ is not the
source. The edge leads to $\ident{v}$ paths: $\sig{w} \leftarrow \sig{w} + (\sig{v} *
\ident{v})$ if $v \neq s$.

The propagation of the dependencies $\ident{w}$ along the edge $\{v,w\}$
should be accounted multiple times as in $\del{v} \leftarrow \del{v} +
\frac{ \sig{v} }{ \sig{w}} \ident{w} ( \del{w} + 1)$.

Finally, for a given source $s$, there are $\ident{s}$ similar
shortest path graphs, and the accumulation of the BC value is $\bcp{w}
\leftarrow \bcp{w} + \ident{s} \del{w}$.

The only path that are ignored in this computation of BC are the paths
between $u \in {\cal I}$ and $v \in {\cal I}$. If ${\cal I}$ is a
type-II identical set, then this path are direct and the computation
of BC 
is correct. However, if ${\cal I}$ is a type-I identical set, these
paths have some impact. Fortunately, it only impacts the direct
neighboor of ${\cal I}$. There are exactly $|{\cal I}| \times ( |{\cal
  I}| - 1)$ paths whose impact is equally distributed among the neighbors of
${\cal I}$.

The technique presented in this section has been presented without
taking ${\tt reach}$ into account. Both techniques can be applied
simultaneously but the details are not presented here due to space
limitation.

\begin{corr}
Eq.~\ref{eq:toprove} is correct after shattering $G$ with articulation
vertices and bridges, and compressing it with degree-1, side, and
identical vertices.
\end{corr}

\subsection{Implementation Details}
There exist linear time algorithms for detecting articulation vertices
and bridges~\cite{Tarjan74,Hopcroft73}. In our implementation of the
preprocessing phase, after detecting all articulation vertices
with~\cite{Hopcroft73}, the graph is decomposed into its biconnected
components at once. Note that the final decomposition is the same when
the graph is iteratively shattered one articulation point at a time as
described above. But decomposing the graph into its biconnected
components is much faster. A similar approach works for bridges and removes
all of them at once. Since the detection algorithms are linear time,
each cumulative shattering operation takes $\mathcal{O}(m + n)$ time.

For compression techniques, detecting recursively all degree-1
vertices takes $\mathcal{O}(n)$ time. Detecting identical vertices
is expected to take a linear time provided a good hash function to compute the
hash of the neighborhood of each vertex. In our implementation, for
all $v \in V_\ell$, we use $hash(v) = \sum_{u \in
  \Gamma(v)}u$. Upon collision of hash values, the neighborhood of the
two vertices are explicitly compared.

To detect side vertices of degree $k$, we use a simple algorithm which
for each vertex $v$ of degree $k$, verifies if the graph induced by
$\Gamma(v)$ is a clique. In practice, our implementation does not search
for cliques of more than $4$ vertices since our preliminary
experiments show that searching these cliques is expensive. Similar to
shattering, after detecting all vertices from a certain type, we apply
a cumulative compression operation to remove all the detected vertices
at once.

\begin{figure}[htbp]
\begin{center}
\includegraphics[width=0.50\textwidth]{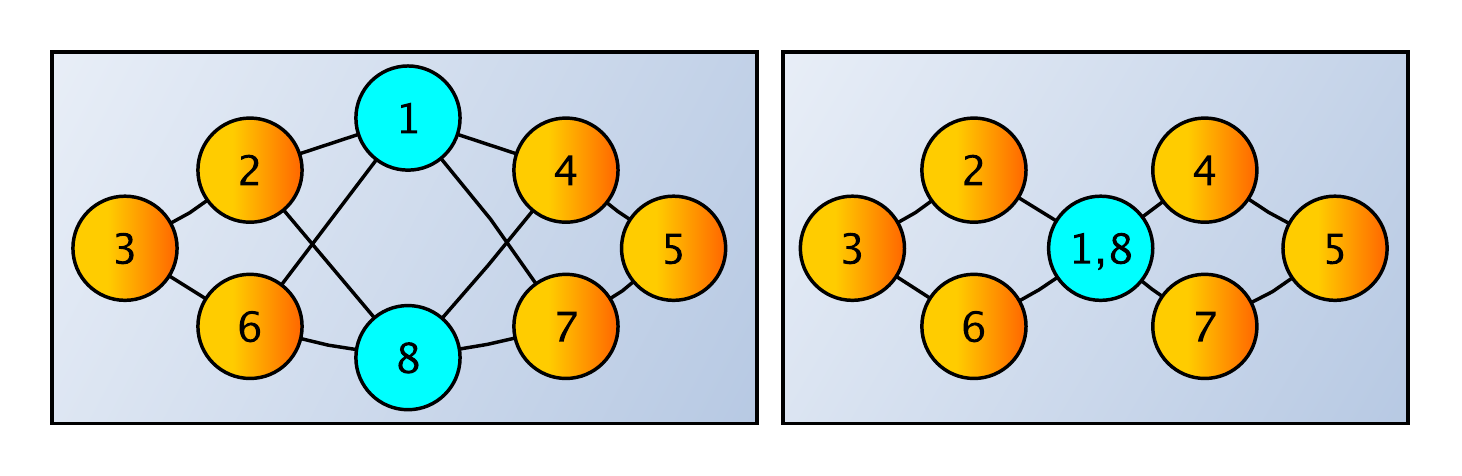}
\caption{Combining identical vertices can create an articulation
  point: Vertices $1$ and $8$ are identical vertices with neighbors
  $\{2,4,6,7\}$. When one of the identical vertices is removed, the
  remaining one is an articulation point. }
\label{fig:idvart}
\end{center}
\end{figure}

The preprocessing phase is implemented as a loop where a single
iteration consecutively tries to shatter/compress the graph by using
the above mentioned five operations.  The loop continues as long as
improvement are made. Indeed, a single iteration of this loop may not
be sufficient since each operation can make the graph amenable to
another one. For example, in our toy graph given in
Figure~\ref{fig:social}.(a), removing the degree-1 vertex Martin makes
Wayne and Henry identical. Furthermore, when Diana is also removed as
a side vertex, Henry and Wayne both become side vertices. Or as
Figure~\ref{fig:idvart} shows, removing identical vertices can form an
articulation vertex .

\section{Experimental Results}\label{sec:exp}

We implemented the original and modified BC algorithms, and the
proposed optimization techniques in {\tt C++}. The code is compiled
with {\tt icc v12.0} and optimization flags {\tt -O2 -DNDEBUG}. The
graph is kept in memory in the compressed row storage~(CRS) format
using $32$-bit data types. The experiments are run on a node with two
Intel Xeon E$5520$ CPU clocked at $2.27$GHz and equipped with $48$GB
of main memory. Despite the machine is equipped with $8$ cores, all
the experiments are run sequentially.

For the experiments, we used $21$ real-life networks from the dataset
of DIMACS Graph Partitioning and Graph Clustering Challenge. The
graphs and their properties are summarized in
Table~\ref{tab:graph_prop}.  They are classified into four
categories. The first one, {\em social}, contains $6$ social
networks. The second one, {\em structural}, contains $5$ structural
engineering graphs. The third one, {\em geographical}, contains $4$
redistricting graphs and one road graph. The last one, {\em misc},
contains graphs from various applications such as autonomous systems,
protein-protein interaction, and power grids.

\begin{table}
\center
\smaller
\begin{tabular}{|l|l|rr|rr|}
\hline
\multicolumn{4}{|c|}{\bf Graph}&\multicolumn{2}{c|}{\bf Time}\\ \hline
application & name & \#vertices & \#edges & org.  & best\\
\hline
\multicolumn{6}{|c|}{Category {\em social}} \\
\hline
Social & CondMat & 16,726 & 47,594 &21.1&  9.1\\
& CondMat03 & 27,519 & 116,181 &102.0&  52.3\\
& hep-th & 8,361 & 15,751 &3.2&  1.6\\
& CondMat05 & 40,421 & 175,691 &209.0&  107.0\\
& PGPgiant & 10,680 & 24,316 &10.7&  3.7\\
& astro-ph & 16,706 & 121,251 &40.3&  22.2\\
\hline
\multicolumn{6}{|c|}{Category {\em structural}} \\
\hline
Auto & bcsstk29 & 13,992 & 302,748 &68.3&  26.4\\
& bcsstk30 & 28,924 & 1,007,284 &399.0&  41.4\\
& bcsstk31 & 35,588 & 572,914 &363.0&  106.0\\
& bcsstk32 & 44,609 & 985,046 &737.0&  77.3\\
& bcsstk33 & 8,738 & 291,583 &37.0&  11.1\\
\hline
\multicolumn{6}{|c|}{Category {\em geographical}} \\
\hline
Redistricting & ak2010 & 45,292 & 108,549 &178.0&  114.0\\
& ct2010 & 67,578 & 168,176 &514.0&  369.0\\
& de2010 & 24,115 & 58,028 &61.4&  40.6\\
& hi2010 & 25,016 & 62,063 &18.4&  12.9\\
\hline
Road & luxembourg & 114,599 & 119,666 &632.0&  390.0\\
\hline
\multicolumn{6}{|c|}{Category {\em misc}} \\
\hline
Router & as-22july06 & 22,963 & 48,436 &39.9&  15.5\\
\hline
Power & power & 4,941 & 6,594 &1.3&  0.7\\
\hline
Biology & ProtInt & 9,673 & 37,081 &11.2&  8.1\\
\hline
Semi-& add32 & 4,960 & 9,462 &1.4&  0.3\\
Conductor& memplus & 17,758 & 54,196 &17.6&  11.2\\\hline\hline
\multicolumn{4}{|r|}{\bf Geomean} & 47.4 & 19.6\\
\hline
\end{tabular}
\caption{Properties of the graphs used in the experiments. Column org.
  shows the original time of $\bcalg$ without any modification. And
  best is the minimum execution time by a combination of the proposed
  heuristics.}
\label{tab:graph_prop}
\end{table}

\subsection{Ordering sparse networks}
As most of the graph-based kernels in data mining, the order of the
vertices and edges accessed by Brandes' algorithm is important. In
today's hardware, cache is one of the fastest and one of the most
scarce resources. When the graphs are big, they do not fit in the
cache, and the number of cache misses along with the number of memory
accesses increases. 

If two vertices in a graph are close, a BFS will access them almost at
the same time. Hence, if we put close vertices in $G$ to close
locations in memory, the number of cache misses will probably
decrease. Following this reasoning, we initiated a BFS from a random
vertex in $G$ and use the queue order of the vertices as their
ordering in $G$. Further benefits of BFS ordering on the execution
time of a graph-based kernel are explained in~\cite{Cong}.

\begin{figure}[htbp]
\center
\includegraphics[width=0.7\columnwidth]{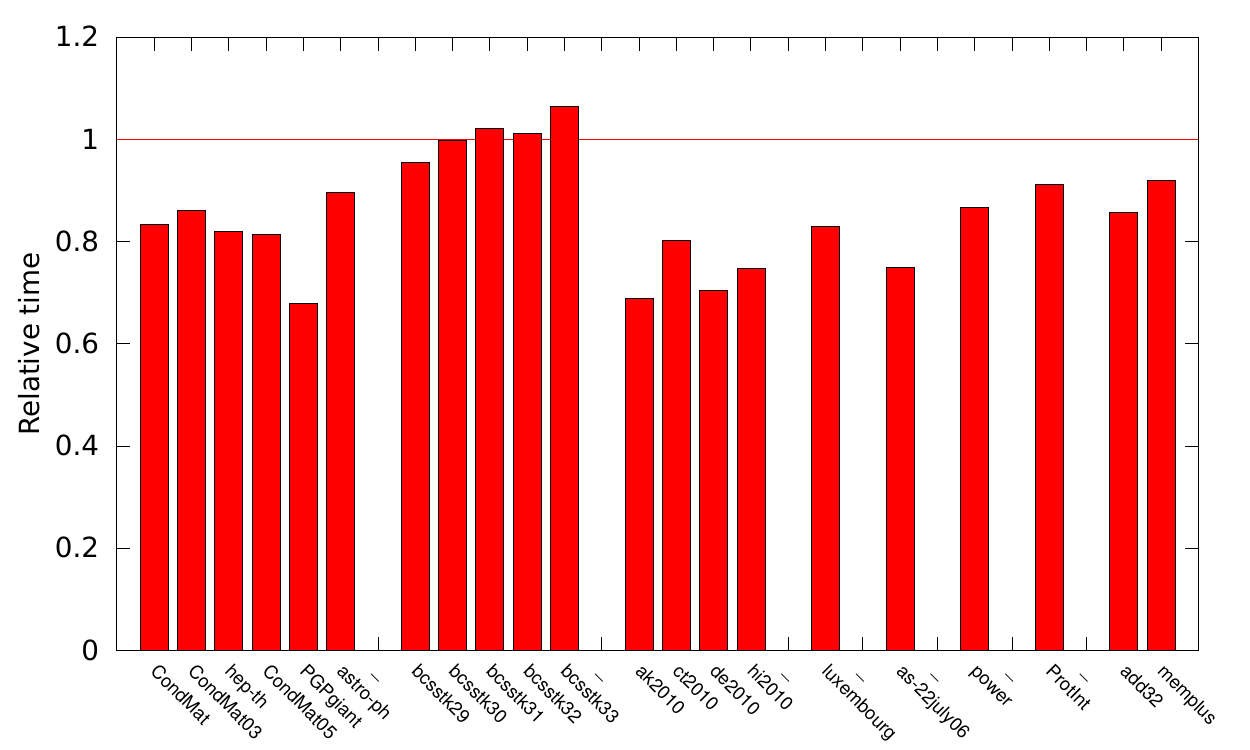}
\caption{Relative performance of BFS ordering with respect to original
time with natural ordering.}
\label{fig:bfs_ord}
\end{figure}

For each graph in our set, Figure~\ref{fig:bfs_ord} shows the time
taken by both the BFS ordering and $\bcalg$ relative to the
original $\bcalg$ execution time with the natural vertex ordering. For
$18$ of $21$ matrices using a BFS ordering 
improved the performance. Overall, it reduced the time to
approximately $80\%$ of the original time on average. Hence compared
with BFS ordering, the natural order of a real-life network has
usually a detrimental effect on the execution time of BC.

\subsection{Shattering and compressing graphs}
For each graph, we tested $7$ different combinations of the
improvements proposed in this paper: They are denoted with {\bf o},
{\bf od}, {\bf odb}, {\bf odba}, {\bf odbas}, {\bf odbai},
and {\bf odbasi}, where 
{\bf o} denotes the BFS {\bf o}rdering,
{\bf d} denotes {\bf d}egree-1 vertices, 
{\bf b} denotes {\bf b}ridge,
{\bf a} denotes {\bf a}rticulation vertices,
{\bf s} denotes {\bf s}ide vertices, and
{\bf i} denotes {\bf i}dentical vertices. The ordering of the letters
denotes the order of application  of the respective improvements.

Given a graph $G$, we measure the time spent for preprocessing $G$ by
a combination to obtain $G'$, computing the BC scores of the vertices in
$G'$, and using these scores computing the BC scores of the vertices in
$G$. For each category, we have two kind of plots:
the first plot shows the numbers of edges in each component of
$G'$. Different components of $G'$ are represented by different
colors. The second plot shows the normalized execution times for all
$7$ combinations. The times for the second chart are normalized
w.r.t. the first combination: the time spent by $\bcalg$ after a BFS
ordering. For each graph in the category, each plot has $7$ stacked
bars representing a different combination in the order described
above.

\begin{figure}[htbp]
\center
\vspace*{-11ex}
\subfigure[Category {\em social}]{
\includegraphics[width=0.53\columnwidth]{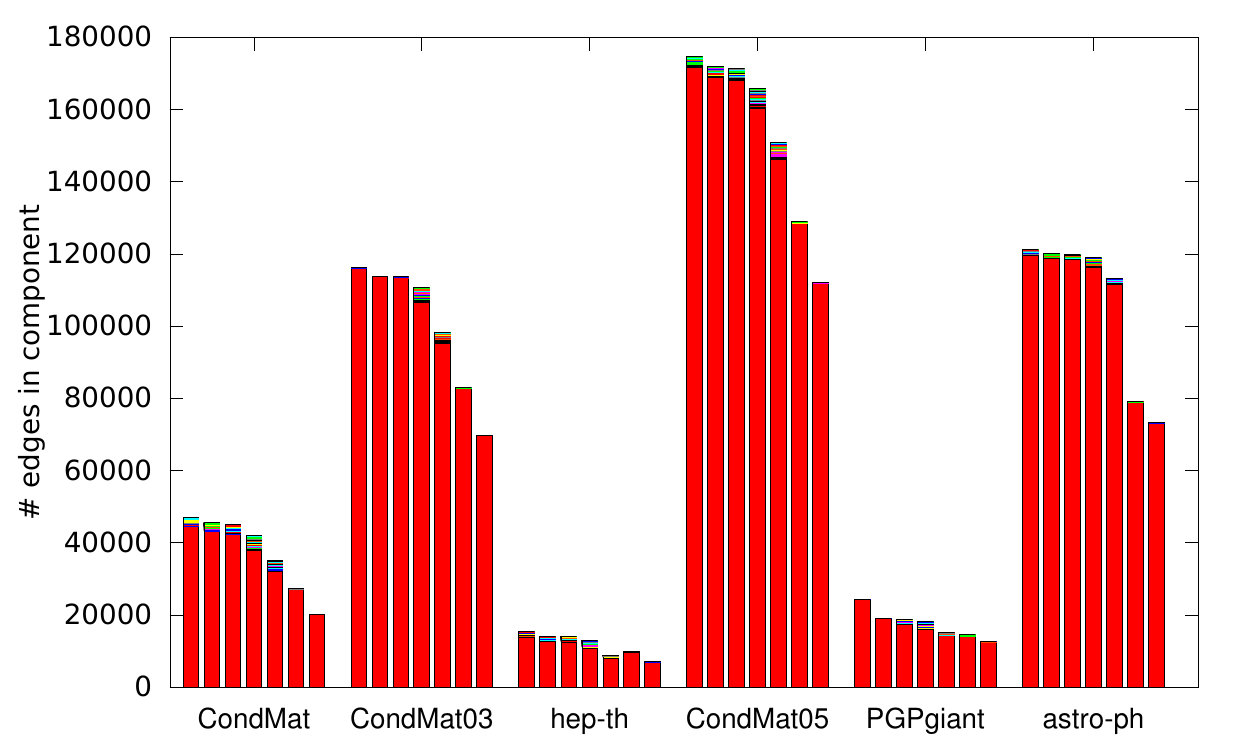}
\includegraphics[width=0.53\columnwidth]{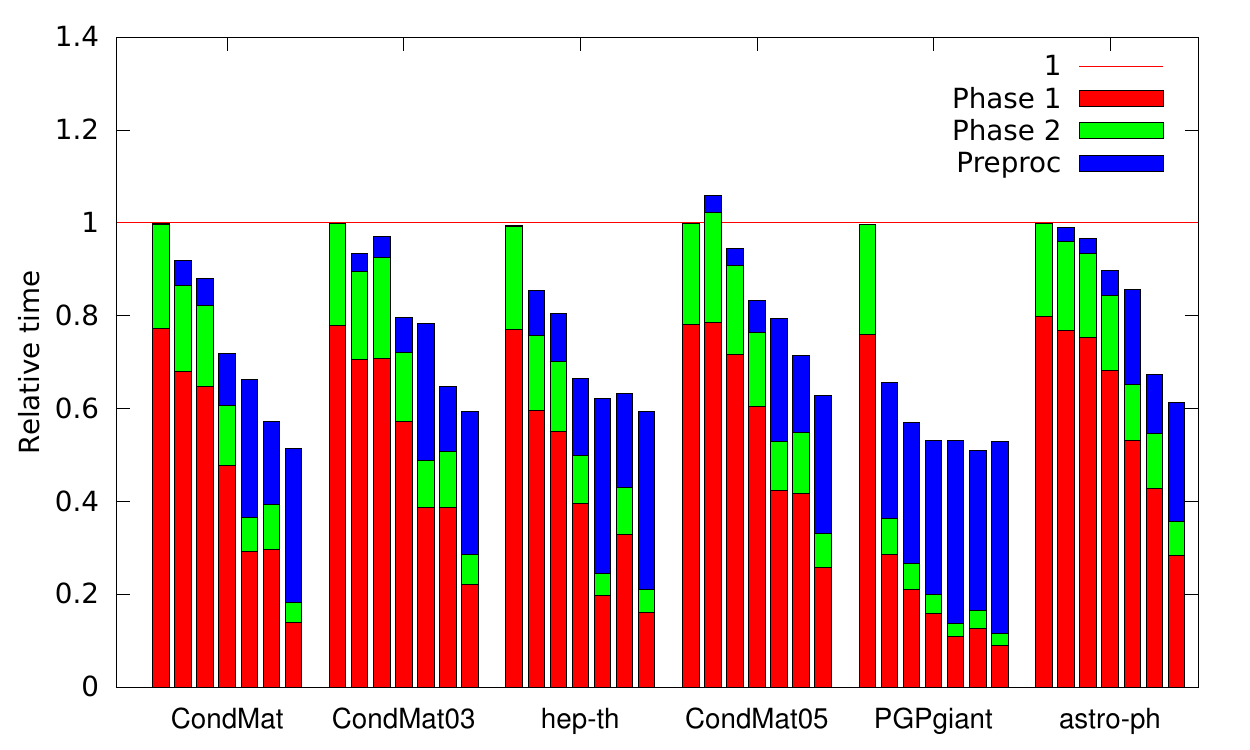}
} \subfigure[Category {\em structural}]{
\includegraphics[width=0.53\columnwidth]{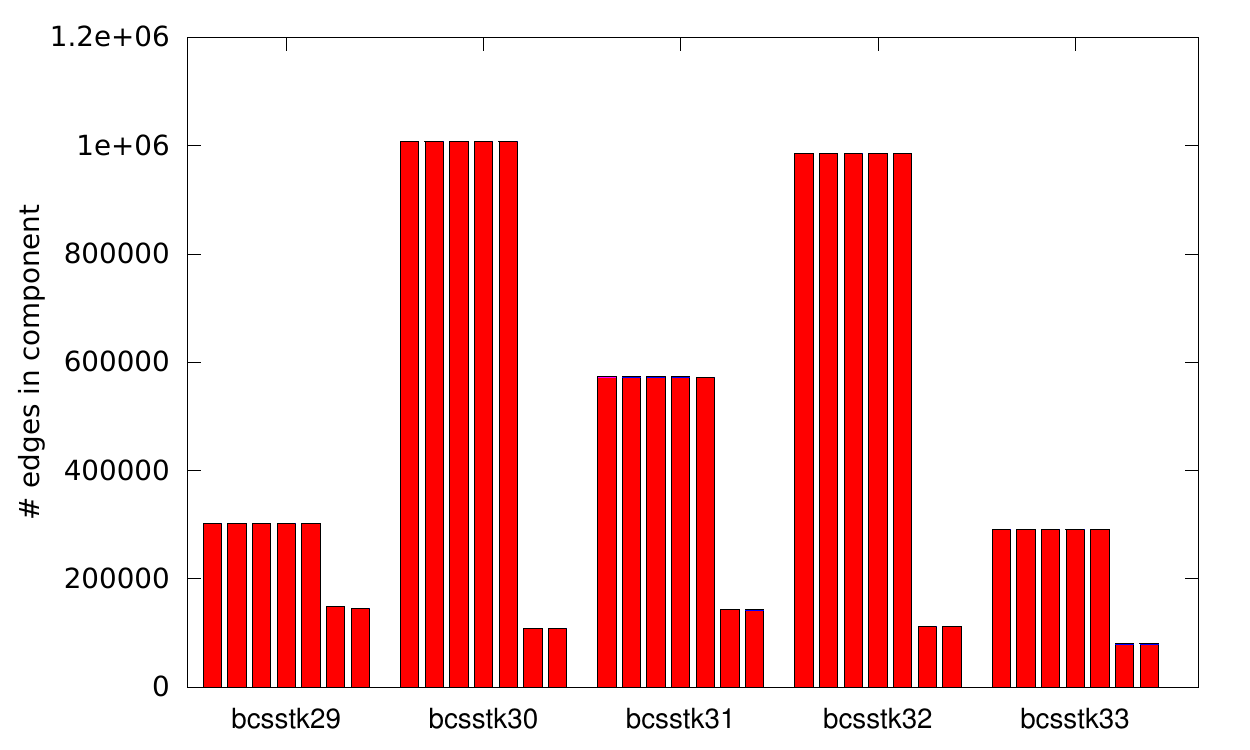}
\includegraphics[width=0.53\columnwidth]{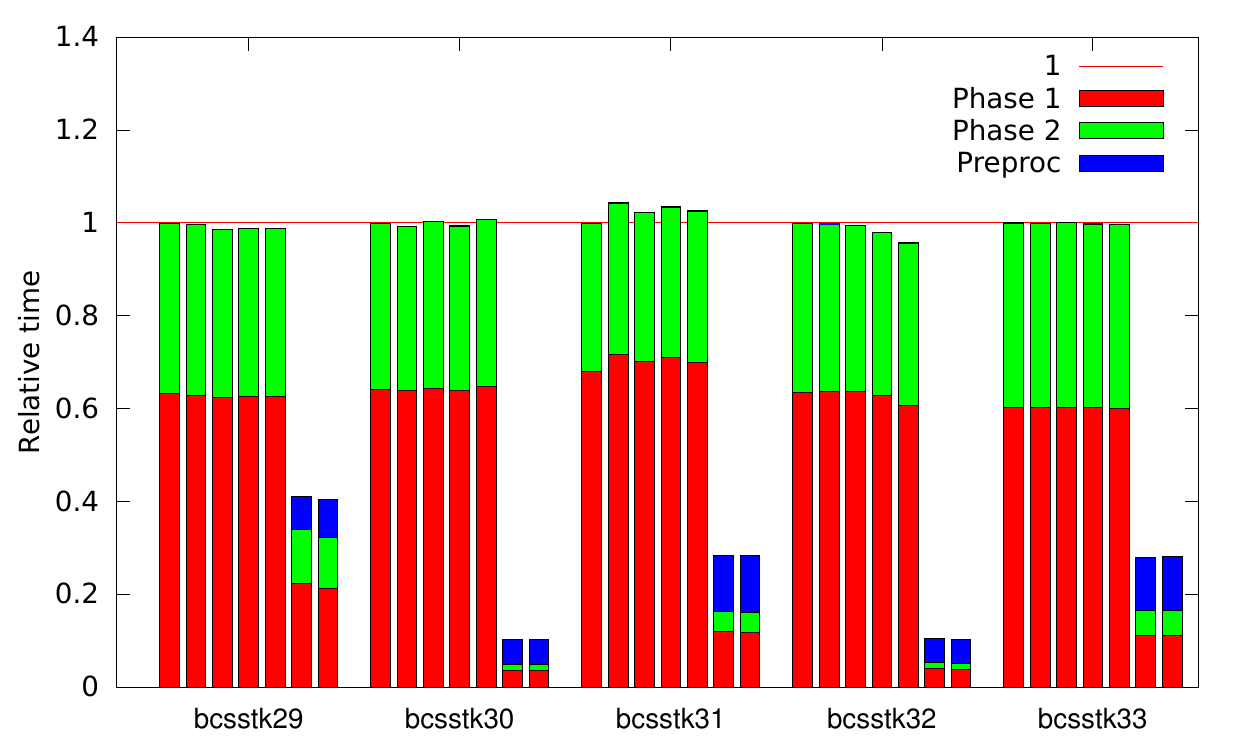}
} \subfigure[Category {\em geographical}]{
\includegraphics[width=0.53\columnwidth]{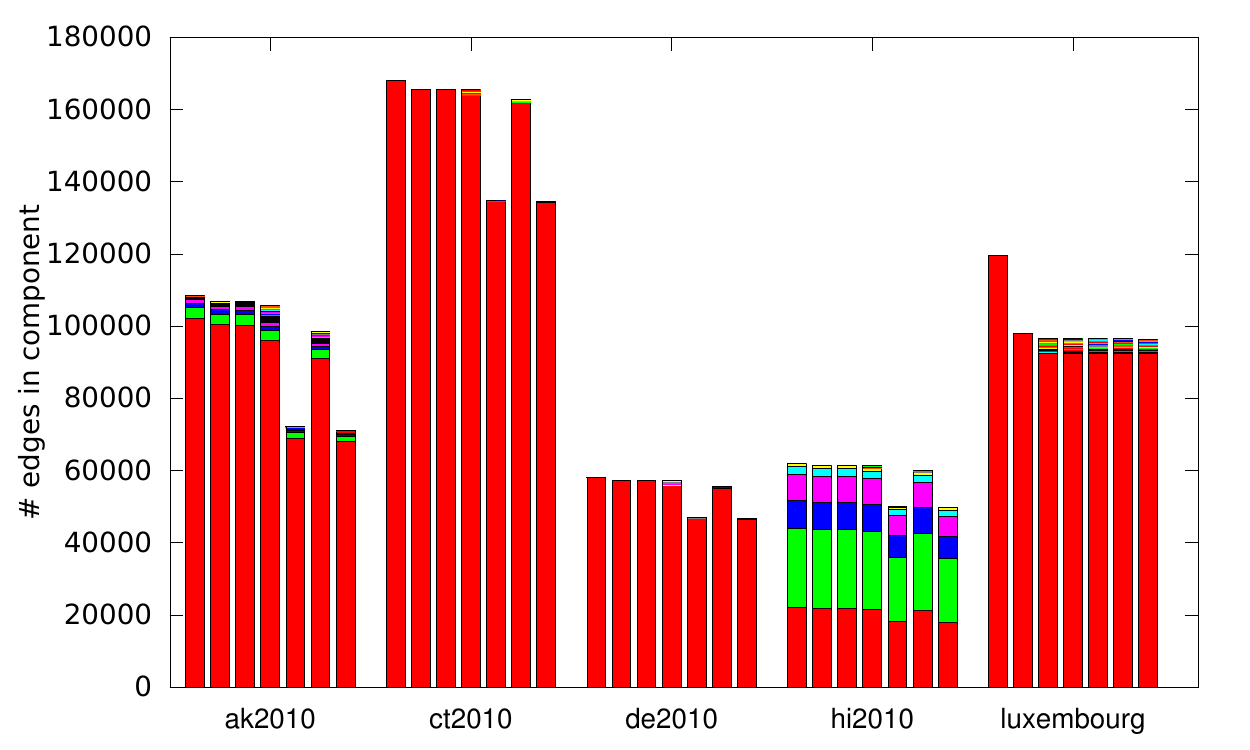}
\includegraphics[width=0.53\columnwidth]{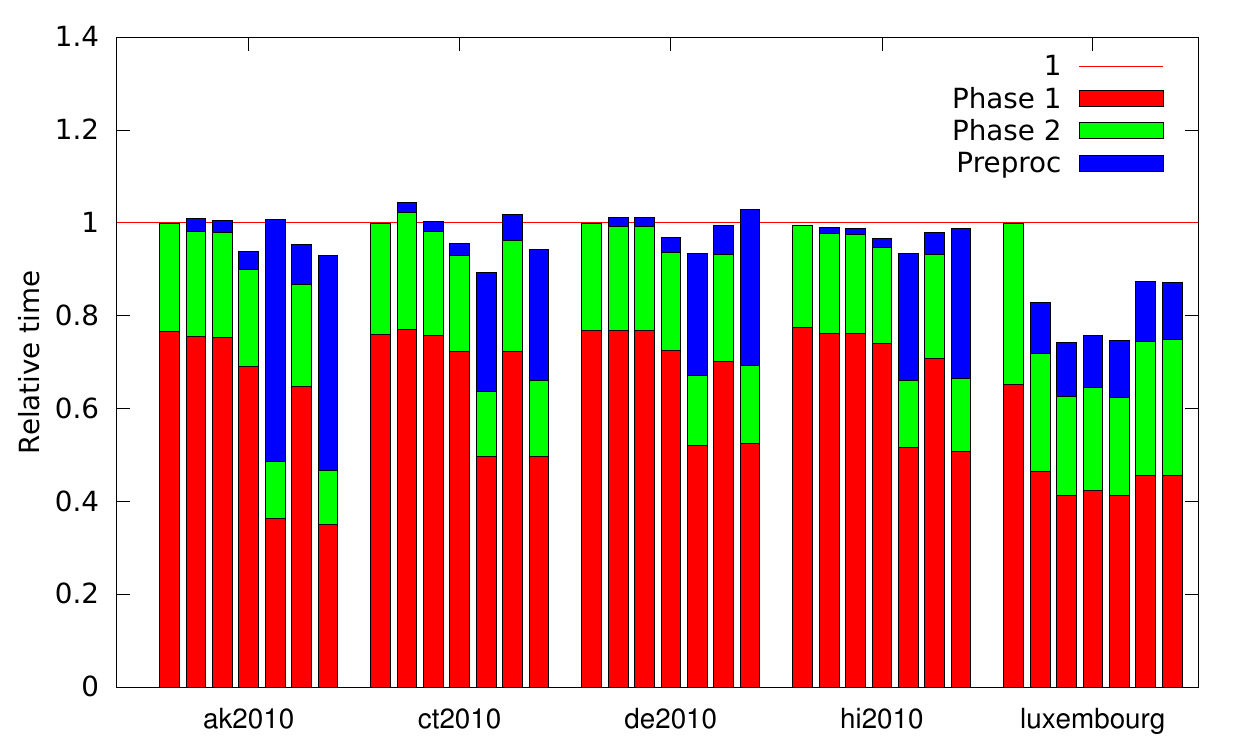}
} \subfigure[Category {\em misc}]{
\includegraphics[width=0.53\columnwidth]{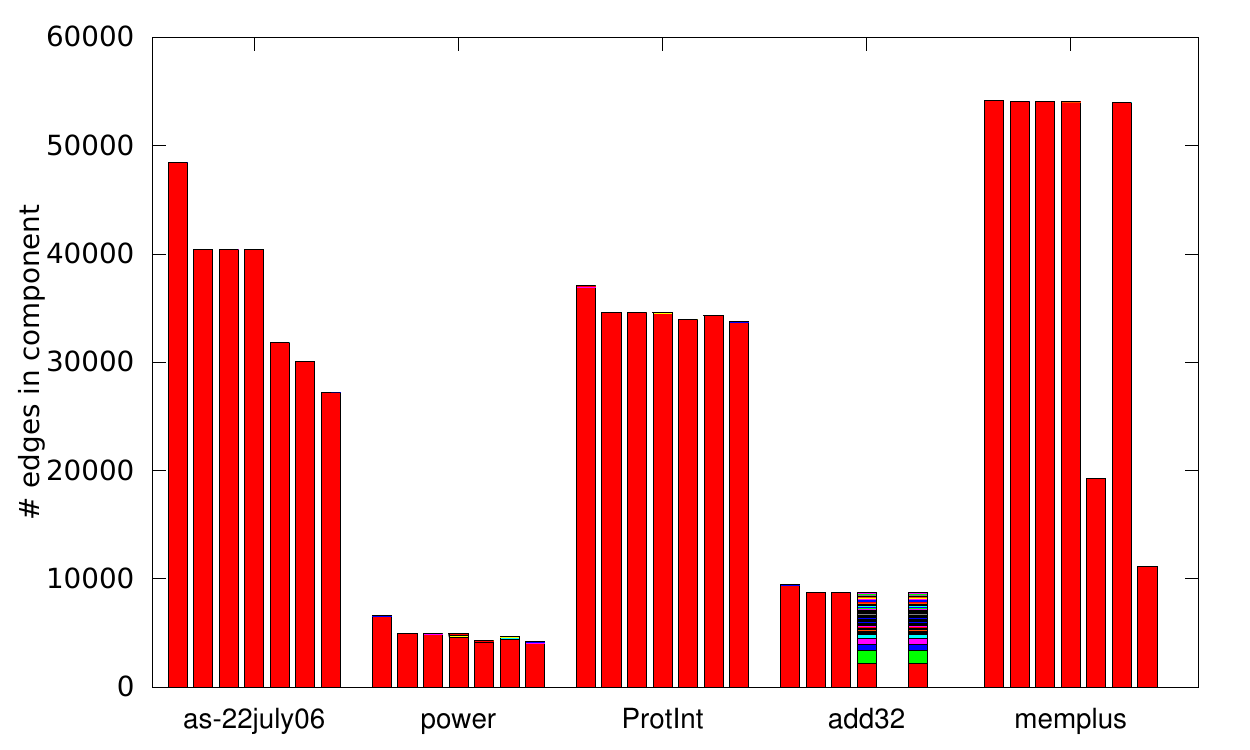}
\includegraphics[width=0.53\columnwidth]{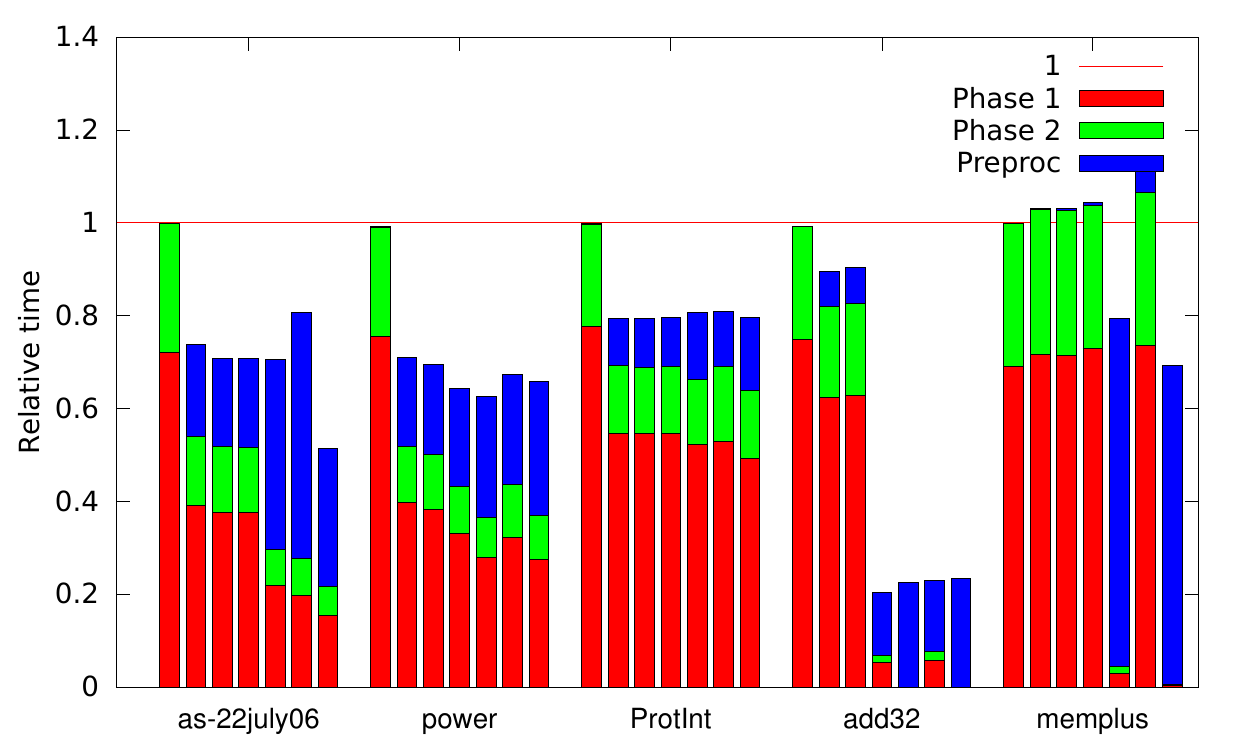}
}\caption{\small Left: The numbers of edges in the connected components of
$G'$ as stack bars. Each component is represented by a different
color. Right: Normalized execution times of preprocessed BC
computations with the combinations {\bf o}, {\bf od}, {\bf odb}, {\bf
odba}, {\bf odbas}, {\bf odbai}, and {\bf odbasi}, respectively, where
the times are normalized w.r.t. {\bf o} and divided to three stages:
preprocessing time and the time spent in the first and second phases
of the BFSs.}
\label{fig:allgraph}

\end{figure}

As Figure~\ref{fig:allgraph} shows, there is a direct correlation
between the remaining edges in $G'$ and the execution time. This
proves that our rationale behind investigating shattering and
compression techniques is valid. However, the figures on the left show
that these graphs do not contain good articulation vertices and
bridges which shatter a graph approximately half. Since, red is almost
always the dominating color, we can argue that
such vertices and edges do not exist in real life graphs.

For {\em social} graphs, each added shattering and compression
technique provides a significant improvement for almost all of the
cases. That is, the best combination is {\bf odbasi} for $5$ out of
$6$ graphs, and the normalized execution time is continuously
decreasing when a combination is enhanced with a new
technique. According to the original and best execution times in Table
~\ref{tab:graph_prop}, for {\em social} graphs, the techniques,
including ordering, provide $53\%$ improvement in total. For {\em
  structural} graphs, although the only working technique is identical
vertices, the improvement is of $79\%$ on the average. For the
redistricting graphs in {\em geographical}, the techniques are not
very useful. However, with the help of BFS ordering, we obtain $32\%$
improvement on average. For the graph {\em luxembourg}, degree-1 and
bridge removal techniques have the most significant impact. Since the
graph is obtained from a road network, this is expected~(roads have
bridges). Hence, if the structure of the graph is known to some
extent, the techniques can be specialized. For example, it is a well
known fact that biological networks usually have a lot of degree-1
vertices but a few articulation vertex. And our results on the graph
{\em ProtInt} confirms this fact since the only significant
improvement is obtained with the combination {\bf od}. In our
experiments, the most interesting graph is {\em add32} since the
combinations {\bf odbas} and {\bf odbasi} completely shatters it. Note
that on the left, there is no bar since there is no remaining edge in
$G'$ and on the right, all the bar is blue which is the color of
preprocessing. When all techniques are combined, we obtain a $59\%$
improvement on average over all graphs.

Please note that the implementation uses $4$ different kernels
depending on whether ${\tt reach}$ and ${\tt ident}$ are used. Each
new attribute brings an increase in runtime which can be seen on {\em
  CondMat03} when going from {\bf o} to {\bf od} and on {\em
  luxembourg} when going from {\bf odba} to {\bf odbai}.

The combinations are compared with each other using a performance
profile graph presented in Figure~\ref{fig:perf_profile}.  A point
$(r,p)$ in the profile means that with $p$ probability, the time of
the corresponding combination on a graph $G$ is at most $r$ times
worse than the best time obtained for that $G$. Hence, the closer
to the y-axis is the better.

\begin{figure}[htbp]
\center
\includegraphics[width=0.7\columnwidth]{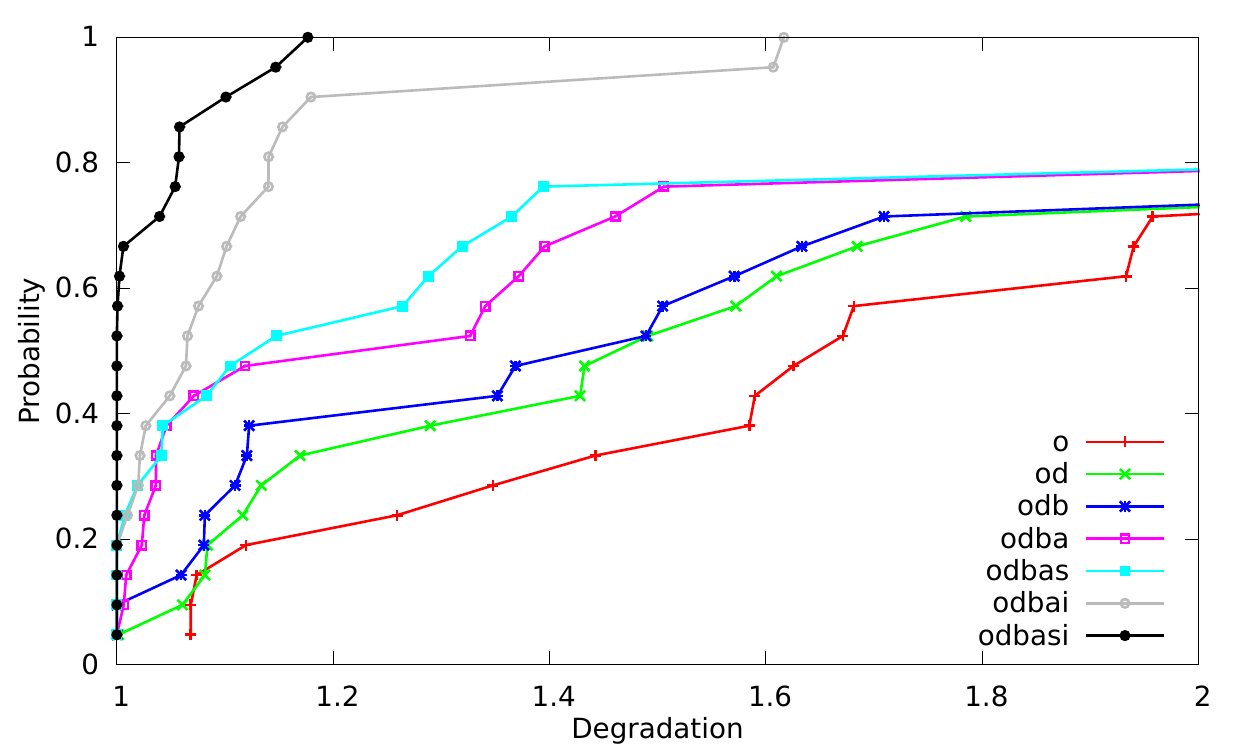}
\caption{Performance profile of various combination of optimization on
  all the selected graphs.}
\label{fig:perf_profile}
\end{figure}

Not using graph shattering techniques ({\bf o}) has the worse performance
profile. It is never optimal. According to the graph, using all
possible techniques is the best idea. This strategy is the optimal one
with more than $60\%$ probability. Clearly, one always wants to use
graph shattering techniques. If little information is available
{\bf odbasi} should be the default choice. However, if one believes
that identical vertices will barely appear in the graph, then
{\bf odbas} might lead to better performances.

\section{Conclusion} \label{sec:con}
 
Betweenness is a very popular centrality metric in practice and proved
to be successful in many fields such as graph mining. But, computing BC
scores of the vertices in a graph is a time consuming task. In this
work, we investigate shattering and compression of networks to reduce
the execution time of BC computations.

The shattering techniques break graphs into smaller components while
keeping the information to recompute the pair and source dependencies
which are the building blocks of BC scores. On the other hand, the
compression techniques do not change the number of components but
reduces the number of vertices and/or edges. An experimental
evaluation with various networks shows that the proposed techniques
are highly effective in practice and they can be a great arsenal to
reduce the execution time while computing BC scores.

We also noticed that the natural order of a real-life network has
usually a detrimental effect on the execution time of BC. In our
experiments, even with a simple and cheap BFS ordering, we managed to
obtain $20\%$ improvement on average. Unfortunately, we are aware of
several works, which do not even consider a simple ordering while
tackling a graph-based computation. So one rule of thumb: ``Order your graphs''.

As a future work, we are planning to extend our techniques to other
centrality measures such as closeness and group-betweenness. Some of
our techniques can readily be extended for weighted and directed
graphs, but for some, a complete modification may be required. We will
investigate these modifications. In addition, we are planning to adapt
our techniques for parallel and/or approximate BC computations.
 




\end{document}